\documentclass[a4paper,superscriptaddress,twocolumn]{revtex4-1}

\usepackage[mode=buildnew]{standalone}

\usepackage{amsmath} 					
\usepackage{amssymb}
\usepackage{amsthm}
\usepackage{array}       				
\usepackage{graphicx}  					
\usepackage{xcolor}       			
\usepackage{makeidx}   		

\usepackage[utf8]{inputenc} 			
\usepackage{url}     					
\usepackage[english]{babel}
\usepackage{siunitx}\usepackage{booktabs}
\usepackage{csquotes}
\usepackage[shortlabels]{enumitem}
\usepackage{braket}
\usepackage{tcolorbox}
\usepackage{tikz}
\usepackage{times}
\usepackage{nicefrac}
\usepackage{booktabs}
\usepackage{mathtools}

\newcommand{\diff}{\ensuremath\mathrm{d}}

\newcommand{\Id}{\ensuremath{\mathbb{I}}}

\newcommand{\e}{\ensuremath{\mathrm{e}}}

\newcommand{\transp}{\ensuremath{\scriptscriptstyle T}}

\newcommand*{\Tr}{\operatorname{Tr}}

\providecommand{\myvec}[1]{\ensuremath{\boldsymbol{#1}}}

\providecommand{\ff}{\ensuremath{\myvec{f}}}

\providecommand{\nn}{\ensuremath{\myvec{n}}}

\providecommand{\rr}{\ensuremath{\myvec{r}}}

\providecommand{\ttheta}{\ensuremath{\myvec{\theta}}}

\providecommand{\mmu}{\ensuremath{\myvec{\mu}}}

\providecommand{\pphi}{\ensuremath{\myvec{\phi}}}
\providecommand{\vvarphi}{\ensuremath{\myvec{\varphi}}}

\providecommand{\calD}{\ensuremath{\mathcal{D}}}

\providecommand{\calF}{\ensuremath{\mathcal{F}}}

\providecommand{\calH}{\ensuremath{\mathcal{H}}}

\providecommand{\calL}{\ensuremath{\mathcal{L}}}
\providecommand{\calM}{\ensuremath{\mathcal{M}}}
\providecommand{\calN}{\ensuremath{\mathcal{N}}}

\providecommand{\calP}{\ensuremath{\mathcal{P}}}

\providecommand{\calR}{\ensuremath{\mathcal{R}}}

\providecommand{\calU}{\ensuremath{\mathcal{U}}}
\providecommand{\calV}{\ensuremath{\mathcal{V}}}

\providecommand{\bbE}{\ensuremath{\mathbb{E}}}

\providecommand{\bbI}{\ensuremath{\mathbb{I}}}

\def\diff{\ensuremath\mathrm{d}}

\newtheorem{observation}{Observation}

\newcommand{\trace}[1]{\Tr\{#1\}}

\AtBeginDocument{
\heavyrulewidth=.08em
\lightrulewidth=.05em
\cmidrulewidth=.03em
\belowrulesep=.65ex
\belowbottomsep=0pt
\aboverulesep=.4ex
\abovetopsep=0pt
\cmidrulesep=\doublerulesep
\cmidrulekern=.5em
\defaultaddspace=.5em
}

\begin{document}
	\title{A variational toolbox for quantum multi-parameter estimation}
	\date{\today}
	
	\author{Johannes Jakob Meyer}
	\affiliation{Dahlem Center for Complex Quantum Systems, Freie Universität Berlin, 14195 Berlin, Germany}
	
	\author{Johannes Borregaard}
	\affiliation{Qutech and Kavli Institute of Nanoscience, Delft University of Technology, 2628 CJ Delft, The Netherlands}
	
	\affiliation{Mathematical Sciences,
    Universitetsparken 5, 2100 K\o{}benhavn \O{}, 
    Matematik E, Denmark}

	\author{Jens Eisert}
	\affiliation{Dahlem Center for Complex Quantum Systems, Freie Universit{\"a}t Berlin, 14195 Berlin, Germany}
	
\begin{abstract}
With an ever-expanding ecosystem of noisy and intermediate-scale quantum devices, exploring their possible applications is a rapidly growing field of quantum information science. In this work, we demonstrate that variational quantum algorithms feasible on such devices address a challenge central to the field of quantum metrology: The identification of near-optimal probes and measurement operators for noisy multi-parameter estimation problems. We first introduce a general framework which allows for sequential updates of variational parameters to improve probe states and measurements and is widely applicable to both discrete and continuous-variable settings. We then demonstrate the practical functioning of the approach through numerical simulations, showcasing how tailored probes and measurements improve over standard methods in the noisy regime. Along the way, we prove the validity of a general parameter-shift rule for noisy evolutions, expected to be of general interest in variational quantum algorithms. In our approach, we advocate the mindset of quantum-aided design, exploiting quantum technology to learn close to optimal, experimentally feasible quantum metrology protocols.
\end{abstract}
\maketitle

Quantum metrology exploits non-classical effects to extend the sensitivity of sensing and parameter estimation methods beyond classical limits. 
The achievable precision in quantum metrology depends on the interplay of quantum correlations of the probe states, the unavoidable quantum noise present in the scheme, the geometry of the parameters to be estimated and the information that is gained by possibly intricate measurements
\cite{giovannetti2006quantum,Demkowicz-Dobrzanski2012}. Carefully designing suitable quantum probes and measurements is therefore a frontier of quantum metrology~\cite{Pezze2020twistingnoiseaway} and can lead to significantly improved sensitivity as has been demonstrated in numerous works~\cite{Aasi2013,Liu2015,Facon2016,Chabuda,Luca2018rmp}. 
Such methods have recently been generalized to a multi-parameter setting where a set of spatially distributed sensors is used to simultaneously estimate a number of distinct parameters or a function thereof~\cite{proctor2018multiparameter,Qian2019,Sekatski2020,Guo2020,xia2020}, setting which is relevant for a broad range of applications ranging from nanoscale NMR~\cite{DeVience2015} to networks of atomic clocks~\cite{Komar2014}.

Designing close to optimal quantum protocols for a specific metrological task is, however, highly challenging. Identifying suitable probes and measurement schemes can be a classically intractable task even in the absence of errors. It requires optimizing over quantum states of high dimension, which soon becomes infeasible in practice due to the exponential growth of dimension with the system size.
To add insult to injury, the solution of this task is even less obvious when realistic constraints such as experimental imperfections and decoherence are taken into account: Such decoherence processes, however, are at the heart of the matter when designing realistic quantum metrological protocols in the first place.
This observation gives rise to the insight that in many relevant and meaningful scenarios, one cannot help but resort to quantum tools to find such protocols.

Fortunately, the increasing body of tools originating from the study of \emph{near-term noisy intermediate scale quantum (NISQ) computers}~\cite{preskill_quantum_2018} may come to the rescue here: This work provides a \emph{variational quantum algorithm} for the optimization of quantum sensing protocols, applicable to a wide range of realistic metrology problems in the general multi-parameter setting. It can be implemented directly on controllable quantum devices, thereby circumventing the need to simulate large quantum systems classically.
Variational quantum algorithms combine hybrid quantum-classical algorithms~\cite{McClean_2016} with variational approaches that have been used in the context of quantum metrology to identify good probes in the single parameter setting~\cite{koczor2019variational} and for generating optimal spin squeezed states in optical tweezer arrays~\cite{kaubruegger2019variational}.

\begin{figure}
    \centering
    \includegraphics[width=\columnwidth]{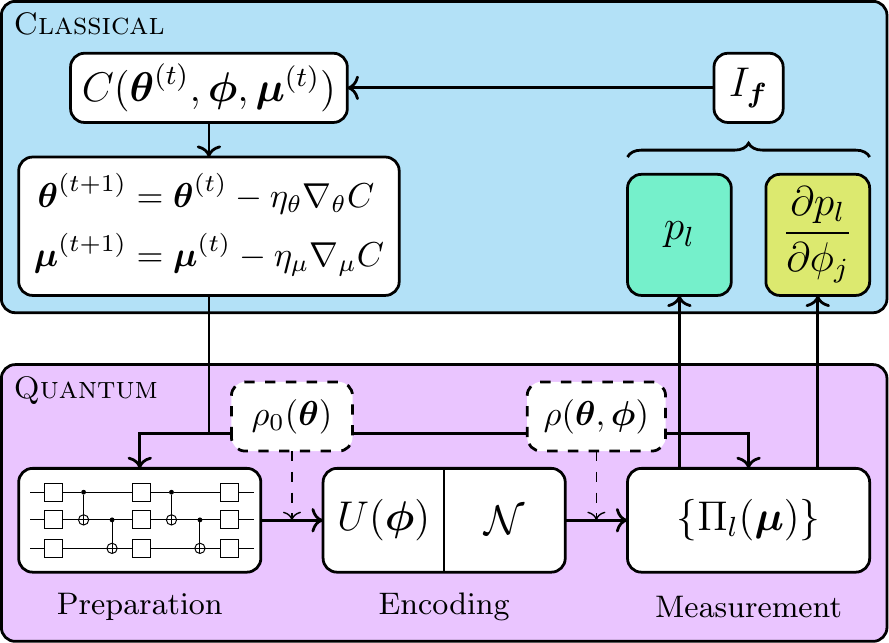}
    \caption{Illustration of the hybrid approach presented in this work to variationally optimize the probe state and measurement for a noisy unitary multi-parameter estimation problem. The various parameters are defined in the main text.}
    \label{fig:process_illustration}
\end{figure}

We consider the generic task of estimating a multi-variate function of the sensing parameters and use the \emph{multi-parameter Cram{\'e}r-Rao bound} to quantify the quality of probe state and measurement as it provides for a saturable lower bound on the quality of any unbiased estimator. 
Importantly, we describe how this is extracted experimentally based on a novel insight regarding the \emph{parameter-shift rule}~\cite{schuld2019evaluating} in a noisy setting, an idea that has arisen in the context of quantum optimal control~\cite{li2017hybrid} and learning applications of quantum circuits \cite{mitarai2018quantum} and which we develop further, expected to be interesting in its own right.
Due to the variational nature of our quantum algorithm -- involving a quantum circuit in its core instead of a possibly inefficient classical prescription -- we can tailor estimation strategies to the specific experimental capabilities and dominating noise sources. Our techniques are very general and can be applied to both discrete variable architectures based on atomic systems~\cite{Bernien2017,Zhang2017} or superconducting qubits~\cite{Arute2019} and continuous variable architectures based on Gaussian light sources~\cite{Larsen2019,Asavanant2019}. 
As such, we add quantum metrology to the possible applications of NISQ devices while at the same time providing a solution to the challenging problem of optimizing quantum metrological prescriptions.

\emph{The structure of the algorithm.}
A high-level description of our algorithm is sketched in 
Fig.~\ref{fig:process_illustration}. In each step of the algorithm, a probe quantum state $\rho_0(\ttheta)$ is generated by a quantum circuit parametrized by $\ttheta$. It undergoes a \emph{noisy unitary} transformation, consisting of a unitary evolution $U(\pphi)$ that encodes the parameters, which are the arguments of $\ff$, the multi-variate function to be estimated, and an arbitrary, possibly non-local, quantum channel $\calN$ accounting for the system noise. 
The state of the system before the measurement is
\begin{align}
    \calN[U(\pphi) \rho_0(\ttheta) U^{\dagger}(\pphi)] = \rho(\ttheta, \pphi).
\end{align}
The subsequent measurement is most generally described by a parametrized positive-operator valued measure (POVM) $\calM = \{ \Pi_l(\mmu) \}$ resulting in measurement output probabilities
\begin{align}
    p_l = \Tr \{ \Pi_l(\mmu) \rho(\ttheta, \pphi) \},
\end{align}
with $\mmu$ being the POVM parametrization. This concludes the part of the protocol that runs on the quantum device.
This step is repeated a number of times to get accurate estimates of the outcome probabilities. These are then used to classically compute a cost function quantifying the estimation quality of the probe state. From this, both the state preparation circuit and the measurement procedure are updated to further increase the estimation quality based on gradient-descent techniques~\cite{sweke2019stochastic}.
The entire procedure is iterated until a minimum is reached, yielding a close to optimal sensing protocol within the variational manifold of probe state preparation and measurement procedure.

\emph{Cost function.}
The central challenge when constructing a variational quantum algorithm is to identify a cost function that captures the nature of the problem and that can be effectively evaluated and differentiated on actual quantum hardware. 
We want to quantify the performance of an estimator for a general multi-variate function $\ff(\pphi)$ of the parameters $\pphi$. We therefore consider the \emph{classical Fisher information matrix (CFIM)}~\cite{lehmann2006theory} $I_{\ff}=J^{\transp} I_{\pphi} J$ where $J$ is the Jacobian of $\ff$ with entries $J_{j,k} = \partial f_j / \partial \phi_k$ and 
\begin{align}\label{eqn:cfim}
    [I_{\pphi}]_{j,k} & \coloneqq \sum_l \frac{(\partial_j p_l)(\partial_k p_l)}{p_l},
\end{align}
where we have used $\partial_j$ as a shorthand notation for $\partial/\partial \phi_j$.
The CFIM gives a fundamental lower bound to the covariance matrix of any unbiased estimator $\hat{\ff}$ of $\ff(\pphi)$ according to the 
\emph{Cram{\'e}r-Rao bound}
\begin{align}\label{eqn:crb}
    \operatornamewithlimits{Cov}[\hat{\ff}] \geq \frac{1}{n}I_{\ff}^{-1},
\end{align}
where $n$ is the number of samples. We note that the Cram{\'e}r-Rao bound can, in principle, always be saturated in the limit $n \to \infty$ by \emph{maximum-likelihood estimation (MLE)} \cite{fisher1922mathematical}.  

For single-parameter estimation, the 
CFIM reduces to a scalar quantity and the 
Cram{\'e}r-Rao bound can directly be used as a cost function. For the multi-variate setting, we will follow the approach of Ref.~\cite{proctor2018multiparameter} and apply a positive semi-definite weighting matrix $W$ to both sides of the Cram{\'e}r-Rao bound \eqref{eqn:crb} and perform a trace to obtain the scalar inequality
\begin{align}\label{eqn:weighted_crb}
    \trace{W \operatorname{Cov}[\hat\ff]} \geq \frac{1}{n}\trace{W I_{\ff}^{-1}} .
\end{align}
The right-hand side is the natural choice for the \emph{cost function}
\begin{align}\label{eqn:cost_function}
    C_W(\ttheta, \pphi, \mmu) &\coloneqq \trace{W I_{\ff}^{-1}}.
\end{align}
While the outcome probabilities $\{p_l\}$ can be readily estimated through repeated measurements, the CFIM also contains their derivatives. 

These derivatives, and the derivatives of the cost function necessary for gradient-based optimization, are computed using the \emph{parameter-shift rule}, which enables the calculation of analytic derivatives of expectation values on quantum hardware~\cite{schuld2019evaluating} for a range of fundamental quantum gates in discrete variable and Gaussian bosonic circuits. Each derivative is computed from the expectation values evaluated at two different points in parameter space.
Recent work~\cite{banchi2020measuring} generalized this approach and showed that analytic gradients for more intricate evolutions can be obtained via stochastic methods, even reaching beyond the noisy unitary model considered in this work. 
While initially derived for unitary quantum circuits, we prove in Appendix~\ref{sec:proof_obs_noisy_parameter_shift} that the parameter-shift rule and its stochastic generalization extend to noisy unitaries. 
As the most common elementary operations considered both in the discrete variable and the Gaussian continuous variable case admit a parameter-shift rule, we will in the following assume that the state preparation and the POVM in our scheme can be \emph{trained} using the parameter-shift rule. 

\begin{figure}
    \centering
    \includegraphics[width=\columnwidth]{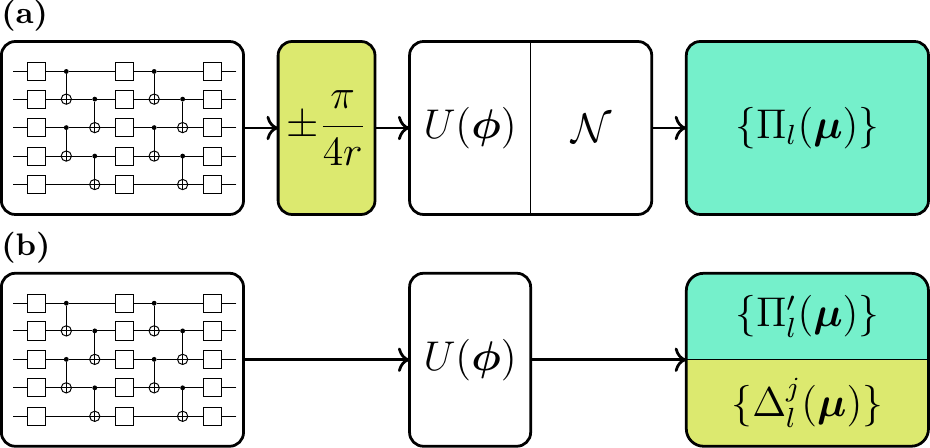}
    \caption{Realization of the approach proposed in Fig.\  \ref{fig:process_illustration}. a) On a sensing platform or during its emulation, the output probabilities are directly computed from the POVM measurements. The derivatives are computed using the parameter-shift rule. b) When simulating on a quantum computer, the action of the noise channel is simulated by altering the POVM. The derivatives are computed by measuring the expectation values of a set of operators, used in training the network.}
    \label{fig:process_realization}
\end{figure}

\emph{Realizing the optimization}.
The difficulty of running the variational algorithm naturally depends on the capabilities of the quantum device available. If we have the opportunity to perform the parametrized state preparation and measurement on the sensing platform, we can use it directly to run the algorithm.
This is of course desirable since it circumvents the need to simulate the encoding and noise channel on another device.   
We can even do without a tomography of the noise channel as only the unitary components
as such need to be known to good precision
and have to be certified and benchmarked \cite{Review}.
The procedure is illustrated in Fig.~\ref{fig:process_realization}(a). 
The output probabilities $\{p_l\}$ are directly accessible from the outputs of the experiment. To obtain estimates of the derivatives $\{\partial_j p_l\}$, it is, however, necessary to be able to re-run the same evolution with the parameter in question shifted by a fixed amount in order to use the parameter-shift rule. Whether this is possible depends on the exact nature of the unitary in the encoding evolution (see Appendix~\ref{sec:app_param_shift_lab} for details).

If the necessary level of control is not achievable on the sensing platform directly, another quantum system with more control can be used to run the algorithm and optimize a sensing strategy.
In particular, NISQ devices with high levels of control could be employed. We stress, however, that the inherent errors of the NISQ device need to be sufficiently low to provide for a faithful emulation of the sensing platform. The noise of the sensing platform can be recreated either through the use of ancillary qubits or repeated sampling of unitaries (see Appendix~\ref{sec:app_emulation}).

If the noise process in question is hard to emulate on the NISQ device, or if the parameter-shift rule cannot be implemented for the encoding process, part of the algorithm can be moved from the quantum device to the classical hardware. This can alleviate the requirements for the NISQ device for a, in many realistic cases, modest overhead on the classical computation. The procedure in this situation is shown in \ref{fig:process_realization}(b). 

We rewrite the output probabilities of the quantum circuit as
\begin{align}
    p_l = \trace{\Pi_l \calN[\rho]} = \trace{\calN^{\dagger}[\Pi_l] \rho} = \Tr\{\Pi_l'\rho \},
\end{align}
where $\calN^{\dagger}$ is the adjoint channel of $\calN$ and we defined the noisy POVM operators $\Pi_l' = \calN^{\dagger}[\Pi_l]$. 

To compute the derivatives of the output probabilities, we use results on unitary estimation problems~\cite{liu2015quantum}. The derivative of the probe after the noiseless part of the evolution $\rho(\pphi) = U(\pphi) \rho_0 U^{\dagger}(\pphi)$ is given by
\begin{align}\label{eqn:derivative_h_ops}
    \partial_j \rho = -i [\calH_j, \rho]
\end{align}
with the Hermitian generators
\begin{align}\label{eqn:def_h_ops}
    \calH_j \coloneqq -i U^{\dagger}(\pphi) (\partial_j U(\pphi)).
\end{align}
For a time evolution under commuting Hamiltonians $H_j$, the generators are identical to the Hamiltonians $\calH_j = H_j$.

Using Eq.~\eqref{eqn:derivative_h_ops}, the derivatives of the output probabilities become
\begin{align}\label{eqn:derivative_probability_errorfree}
    \partial_j p_l &= \trace{\Pi_l' \partial_j \rho} 
    = -i \trace{\Pi_l' \calH_j \rho - \Pi_l' \rho \calH_j}\\
    &= -i \trace{[\Pi_l', \calH_j]\rho}.\nonumber
\end{align}
Defining the derivative operators
\begin{align}\label{eqn:def_derivative_ops}
    \Delta_l^j \coloneqq -i [\Pi_l', \calH_j]
\end{align}
allows us to recast the CFIM as a function of expectation values, similar in form to Eq.~\eqref{eqn:cfim}:
\begin{align}\label{eqn:cfim_expectation_values}
    [I_{\pphi}]_{j,k} = \sum_l \frac{\langle \Delta_l^j \rangle \langle \Delta_l^k \rangle}{\langle \Pi_l' \rangle},
\end{align}
where all expectation values are evaluated on the noise-free encoded state $\rho(\pphi)$.

We note that the operators $\{\Pi_l'\}$ and $\{\Delta_l^j\}$ have to be precomputed and stored classically in this hybrid algorithm. This requires an efficient classical representation of both the error channel and the POVM. This is feasible in practice, as the POVMs usually considered have unit rank. The rank of the noisy POVM operators is then maximally the Kraus rank of $\calN$. Locality of measurements and noise further simplifies the representation of the POVM operators. If necessary, they can also be stored in a format of \emph{tensor networks}~\cite{Chabuda}.

Because the CFIM needs to be evaluated on the noise-free encoded state, we need to parametrize the POVM $\{\Pi_l(\mmu)'\}$ classically. If the POVM is parametrized in terms of quantum operations, we can again use the parameter-shift rule to compute the derivatives with respect to the parameters $\{\mu_k\}$.

\emph{Numerical experiments.}
To showcase possible applications of our variational approach, we numerically investigate two exemplary noisy estimation problems. The experiments have been implemented using the PennyLane library for quantum machine learning~\cite{bergholm2018pennylane} and the QuEST mixed state simulator~\cite{jones2019quest}.

We first employ the setting of Ramsay spectroscopy, a widely used technique for quantum metrology with atoms and ions. The metrological parameters are phase shifts $\pphi$ arising from the interaction of probe ions modeled as two-level systems with an external driving force. We follow Ref.~\cite{huelga1997improvement} and model the noise in the parameter encoding as local dephasing with dephasing probability $p$. We consider a pure probe state and a projective measurement, where the computational basis is parametrized by local unitaries.
\begin{figure}
    \centering
    \includegraphics{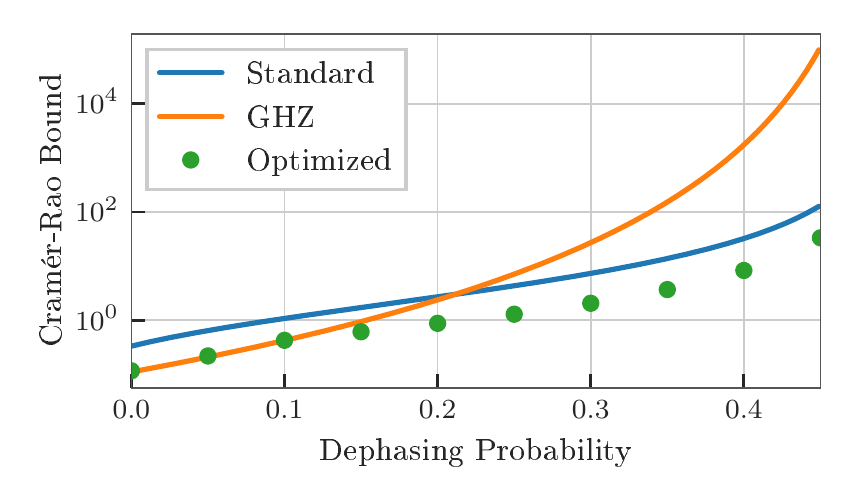}
    \caption{Semi-logarithmic plot of the Cram{\'e}r-Rao bound for average phase sensing over the dephasing probability of (blue) the standard Ramsey probe and (orange) a GHZ probe. Each green dot marks the bound for a probe optimized for the particular dephasing probability. The phase shifts were fixed at the optimal spot for GHZ sensing at $\phi_j = \frac{\pi}{6}$.
    We recover the optimality of the GHZ probe in the noiseless case and find improved sensing protocols in the noisy case.}
    \label{fig:ramsay_results}
\end{figure}
We optimize a sensing strategy for a system of three ions and instead of estimating the phases independently, we analyze the task of estimating their average $\frac{1}{N}\sum_j \phi_j$. We first reproduce known results to validate the performance of our approach. Ref.~\cite{proctor2018multiparameter} has shown that generalized GHZ states are optimal when the encoding is noise-free. We have emulated the sensing task with $\phi_j = \frac{\pi}{6}$, the sweet spot for GHZ sensing and performed co-optimization of state preparation and measurement.
The parametrization was chosen so that it can produce any three qubit quantum state and any local measurement. Fig.~\ref{fig:ramsay_results} shows that we recover the optimality of the GHZ sensing procedure with a measurement in the Hadamard basis for $p=0$. At increasing noise levels, the advantage of GHZ sensing disappears as expected~\cite{huelga1997improvement} and we find sensing procedures that outperform both GHZ and standard Ramsay spectroscopy.

As a second task, we will apply our algorithm to the setting of spin imaging. This is at the heart of nanoscale NMR, which has a wide range of applications within chemistry and biological imaging. In particular, nanoscale sensors based on \emph{nitrogen-vacancy (NV) centers} have shown great potential~\cite{DeVience2015,Shi2015,Aslam2017,Lovchinsky2016,Schmitt2017} and we will consider the task of determining the position of a spin by triangulation with three NV center probes. In short, the dipole-dipole interaction between the NV centers and the spin shifts the energy levels of the NV centers resulting in a position dependent phase shift. From measurements of this phase, the position of the spin can be determined. We neglect the dipole-dipole interaction between the three NV centers assuming that they can be decoupled be appropriate pulse sequences~\cite{Abobeih2019}. Furthermore, we  perform a secular approximation, simplifying the interaction to a Pauli rotation about the symmetry axis of the NV centers~\cite{sushkov2014magnetic}. We model the encoding noise as local dephasing with a fixed probability of $p_e = \num{0.1}$.
We refer to Appendix~\ref{sec:app_numerics} for a more detailed derivation of our model.
\begin{figure}
    \centering
    \includegraphics{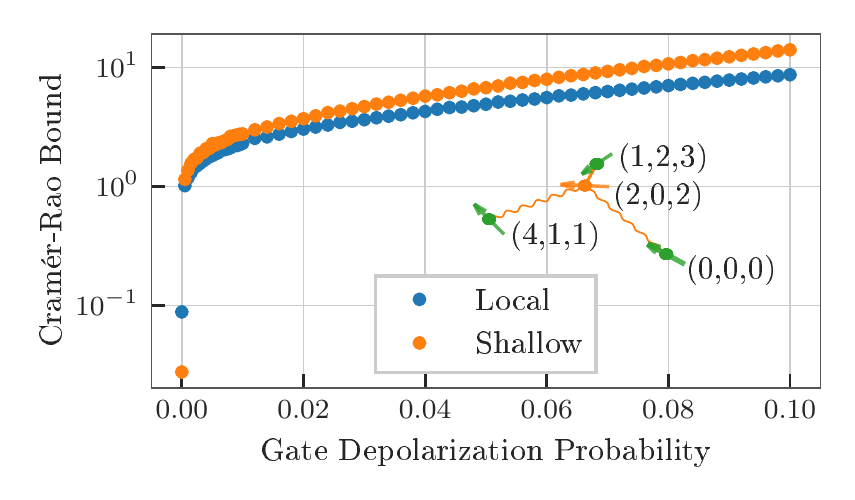}
    \caption{Semi-logarithmic plot of the weighted Cram{\'e}r-Rao bound for NV trilateration over the gate depolarization probability of (blue) local probes and (orange) shallow entangled probes. Each dot marks the bound for a probe optimized for the particular gate depolarization probability. The measurement is always local. The inset shows the positions of the NV centers and the target spin used in the experiment. The entangling operations only yield an advantage for vanishing gate noise.}
    \label{fig:nv_results}
\end{figure} 
We use our algorithm to study the influence of gate noise, which we model as depolarizing noise with depolarization probability $p_g$, on the quality of the sensing protocol. We compare a local probe to a probe that has two additional entangling operations while the measurement is always local. The exact structure of the sensing protocols is detailed in Appendix~\ref{sec:app_numerics}. Fig.~\ref{fig:nv_results} shows that the shallow entangled probe yields an advantage in the absence of gate noise. But even at very small gate noise levels, the additional noise from the entangling operations will make the entangled probe inferior to the local probe. Although we only consider one specific ansatz for an entangled probe, this sheds light on the importance of gate noise in realistic probe preparations.

\emph{Summary and outlook.}
In this work, we have introduced a variational approach that can be applied to a wide variety of metrological problems to optimize multi-parameter estimation schemes both on sensing platforms and on quantum devices. These approaches 
can be tailored to specifically match the available hardware and the sensing problem in question. Instead of addressing the challenges of noisy quantum metrology by classically computing optimal probes and measurements, we have shown how to learn variational parameters of a quantum circuit
for this purpose. The machinery used can also be
used to improve upon an already good guess.

While the approach outlined in the main text is widely applicable, covering arbitrary noise models and post-processing functions, it still does not cover all possible metrological tasks. In Appendix~\ref{sec:app_ext}, we therefore present extensions of our algorithm that take into account prior knowledge about the underlying parameters, 
handle mutual time-dependence of unitary evolution and noise, 
and outline how we can benchmark our results against the ultimately attainable precision given by 
the \emph{quantum Cram{\'e}r-Rao bound}.
The results of Appendix~\ref{sec:proof_obs_noisy_parameter_shift} furthermore allow the extension of our method to the metrology of error channel parameters.
We hope that the present work contributes
to a shift in mindset -- somewhat reminiscent of the idea of
quantum computer-aided design of physical platforms \cite{kyaw2020quantum} --
so that not all sensing protocols have to be designed classically, but that one can make use of quantum technologies themselves to actually learn improved quantum protocols. 

\emph{Acknowledgements.} This work has been supported by
the BMWi (PlanQK), the MATH+ excellence cluster (EF1-7), the Templeton Foundation, and the DFG (CRC 183, B01). JB acknowledges funding from the NWO Gravitation Program Quantum Software Consortium. The authors endorse Scientific CO\textsubscript{2}nduct~\cite{conduct} and provide a CO\textsubscript{2} emission table in the appendix.

\emph{Author contributions.} J.~J.~M.~conceived the method and performed the numerical experiments. J.~B.~and J.~E.~supported research and development. All authors contributed to the manuscript.

%

\clearpage
\appendix
\onecolumngrid

\section{Analytic gradients for noisy quantum circuits}\label{sec:proof_obs_noisy_parameter_shift}
Variational quantum algorithms (VQAs) \cite{mcclean2016theory} use a classical feedback loop to train parametrized quantum circuits to optimize a given cost function encoding the problem of interest. They are widely considered among the most promising applications of noisy intermediate scale quantum (NISQ) devices. 
The efficient training of a VQA is key to a successful application of these algorithms. While gradient-free methods like finite-differences and simultaneous perturbation stochastic approximation (SPSA) can, in principle, be employed, Ref.~\cite{harrow2019low} has shown that gradient-based methods give rise to a provable advantage. Due to the importance of VQA optimization, the development of bespoke gradient descent methods is an active field of research with a multitude of proposals for such methods.

A workhorse of gradient-based methods is the \emph{parameter-shift} rule~\cite{li2017hybrid,mitarai2018quantum,schuld2019evaluating,sweke2019stochastic}, which enables the calculation of analytic gradients of expectation values on actual quantum hardware. This is possible for a wide range of relevant gates by evaluating the same circuit at shifted points in parameter space as
\begin{align}
    \partial_{\mu} \langle H \rangle(\mu) = r\left[ \langle H \rangle\left(\mu + \frac{\pi}{4r}\right) - \langle H \rangle\left(\mu - \frac{\pi}{4r}\right) \right],
\end{align}
where $r$ is a constant depending on the gate parametrized by $\mu$.
Recent work \cite{banchi2020measuring} has extended the reach of the parameter-shift rule to more intricate evolutions using stochastic methods.

In this appendix -- which can actually be read as a standalone note interesting
its own right -- we will confirm some folklore about the parameter-shift rule, namely that it can also be used to train \emph{noisy quantum circuits} under a reasonable noise model. We note that this is only a small step from the initial work~\cite{li2017hybrid}, but a step that we here make fully explicit. We will also note that convex combinations of quantum channels, a class comprising prominent examples like dephasing and depolarizing channels, admit a parameter-shift rule that can be of use in error channel simulations.

\emph{General prerequisites for a parameter-shift rule.}
We first reformulate the parameter-shift rule in terms of quantum channels. Let $\calN(\mu)$ be a quantum channel parametrized by $\mu$. The central requirement for a parameter-shift rule is that the channel can realize its own derivative as a linear combination
\begin{align}\label{eqn:linear_combination_gradient}
    \partial_{\mu} \calN(\mu)[\rho] = \sum_j c_j \calN(g_j(\mu))[\rho].
\end{align}
Under this assumption, any expectation value
\begin{align}
    f(\mu) = \Tr \{ H \calN(\mu)[\rho] \}
\end{align}
can be differentiated as
\begin{align}
    \partial_{\mu} f(\mu) &= \Tr \{ H \partial_{\mu}\calN(\mu)[\rho] \}   \\
    &= \sum_j c_j \Tr \{ H \calN(g_j(\mu))[\rho] \}  \nonumber \\
    &= \sum_j c_j f(g_j(\mu)). \nonumber
\end{align}
Because we have left both the input state $\rho$ and the observable $H$ arbitrary, the parameter-shift rule also holds if the channel is wedged between other channels that do not depend on the same parameter. We can simply attribute the preceding channel to the quantum state and the succeeding channel to the observable via its adjoint, to get
\begin{align}
    f(\mu) &= \Tr \{ H' (\calL \circ \calN(\mu) \circ \calR)[\rho'] \} \\
    &= \Tr \{ \calL^{\dagger}(H') \calN(\mu)[\calR[\rho']] .\nonumber
\end{align} 
This reasoning also holds for circuits with multiple parameters, as long as all parameters act in independent quantum channels. In this case, the channels that are independent of $\mu$ can be absorbed into $\calL$ and $\calR$.

When optimizing a VQA, unbiased estimates of the gradient are sufficient to guarantee convergence~\cite{sweke2019stochastic}. This spirit has successfully been taken up by Ref.~\cite{banchi2020measuring}. They exploit the fact that the gradient of a unitary evolution can be expressed as an integral, effectively yielding a parameter-shift rule with an infinite number of terms. Still, an unbiased estimator can be constructed when using Monte-Carlo integration to estimate the integral. As the integral representation is equivalent to Eq.~\eqref{eqn:linear_combination_gradient}, the above reasoning also holds for the methods developed in Ref.~\cite{banchi2020measuring}

\emph{Channels admitting a parameter-shift rule.}
The influential work of Ref.~\cite{schuld2019evaluating} has proven that a unitary evolution under a Hamiltonian with two distinct eigenvalues fulfills the above requirement for a parameter-shift rule:
\begin{observation}[Gradient as a linear combination for unitaries]\label{obs:unitary_parameter_shift}
    A unitary evolution of the form
    \begin{align}
        \calU(\mu)[\rho] &= \e^{-i\mu G} \rho \e^{i \mu G},
    \end{align}
    where the spectrum of $G$ only contains two distinct eigenvalues $\operatorname{spec}(G) = \{ \lambda_1, \lambda_2 \}$ can express its own gradient as the linear combination 
    \begin{align}
        \partial_{\mu}\calU(\mu)[\rho] =  r \left[\calU\left(\mu+\frac{\pi}{4r}\right)[\rho] - \calU\left(\mu-\frac{\pi}{4r}\right)[\rho]\right],
    \end{align}
    where $r = |\lambda_1 - \lambda_2|/2$.
\end{observation}
For the sake of completeness, a proof of this observation is given at the end of this section.
The above observation is valid for any unitary evolution that is generated by a Pauli word. As the Pauli words generate the unitary group, any unitary transformation has a parametrization that supports a parameter-shift rule.
The reach of the parameter shift rules is extended beyond these simple generators in Ref.~\cite{banchi2020measuring}, thereby allowing to estimate gradients for more complicated evolutions on quantum hardware. 
We note that Ref.~\cite{schuld2019evaluating} also provides a parameter-shift rules for all elementary Gaussian bosonic operations, which means that any Gaussian evolution build from these operations can be parametrized in a way that supports a parameter-shift rule.

This work considers \emph{noisy unitaries}, consisting of a perfect unitary evolution followed by an arbitrary, but parameter-independent, noise channel,
\begin{align}
    \calV(\mu) = \calN \circ \calU(\mu).
\end{align}
Typically, the noise channel will be depolarizing or dephasing. We can also model lossy optical components where the perfect operation is followed by a pure loss channel. Another interesting application is the modeling of imperfect control. Assume that we perform a Hamiltonian evolution $U(\mu) = \e^{-i \mu G}$ but due to imperfect control, we are actually implementing $U(\mu + \chi)$ with probability $p(\chi)$. This can be formalized by a noisy unitary with noise channel
\begin{align}
    \calN(p)[\rho] = \int \diff \chi \, p(\chi) U(\chi) \rho U^{\dagger}(\chi)
\end{align}
We already argued that noise channels following the evolution in question do not impede the application of a parameter-shift rule. This means that noisy unitaries as building blocks can be differentiated by the same parameter-shift rule as the perfect unitary.

We will now move on to show that parameter-shift rules are not limited to unitary evolutions alone. Convex combinations of channels also admit a parameter-shift rule that can be exploited in simulation and emulation contexts.
We will consider the channel
\begin{align}
    \calN(p)[\rho] = (1-p) \calN_1[\rho] + p \calN_2[\rho],
\end{align}
where $\calN_1$ and $\calN_2$ are arbitrary quantum channels of compatible dimensions. Many interesting channels can indeed be captured as a convex combinations of two channels. Consider, as examples, the dephasing channel
\begin{align}
    \calN(p)[\rho] = (1-p) \rho + p Z \rho Z
\end{align}
and the depolarizing channel
\begin{align}
    \calN(p)[\rho] &= (1-p) \rho + p \left(\frac{1}{3} X \rho X + \frac{1}{3} Y \rho Y + \frac{1}{3} Z \rho Z\right).
\end{align}
We will now show that evaluating the channel at two different values of $p$ is sufficient to compute the gradient:
\begin{observation}[Gradient as a linear combination for convex combinations of channels]\label{obs:channel_interpolation_shift}
    An interpolated quantum channel of the form
    \begin{align}
        \calN(p)[\rho] = (1-p) \calN_1[\rho] + p \calN_2[\rho].
    \end{align}
    can express its own gradient independent of $p$ as the linear combination
    \begin{align}
        \partial_p \calN(p)[\rho] &= \frac{1}{q_1 - q_2}\left( \calN(q_1)[\rho] - \calN(q_2)[\rho] \right)
    \end{align}
    for all distinct $q_1, q_2 \in [0, 1]$.
\end{observation}
\begin{proof}
We rephrase the application of $\calN$ in vector notation:
\begin{align}
    \calN(p)[\rho] = \begin{pmatrix} 1-p \\p \end{pmatrix} \begin{pmatrix} \calN_1[\rho] \\ \calN_2[\rho] \end{pmatrix},
\end{align}
where the implicit product of vectors is the scalar product. The derivative is given by
\begin{align}
    \partial_p \calN(p)[\rho] = \begin{pmatrix} -1 \\ 1 \end{pmatrix} \begin{pmatrix} \calN_1[\rho] \\ \calN_2[\rho] \end{pmatrix}.
\end{align}
Now, assume we have access to two indepenent realizations of the channel at $q_1$ and $q_2$. We then get the vector of outcomes
\begin{align}
    \begin{pmatrix} \calN(q_1)[\rho] \\ \calN(q_2)[\rho]\end{pmatrix} &= \begin{pmatrix}
    1-q_1 & q_1 \\ 1-q_2 & q_2 \end{pmatrix}\begin{pmatrix} \calN_1[\rho] \\ \calN_2[\rho] \end{pmatrix}.
\end{align}
From this, we can calculate the derivative using the matrix inverse
\begin{align}
    \partial_p \calN(p)[\rho] &= \begin{pmatrix} -1 \\ 1 \end{pmatrix} \begin{pmatrix} \calN_1[\rho] \\ \calN_2[\rho] \end{pmatrix} \\
    &=\begin{pmatrix} -1 \\ 1 \end{pmatrix}\begin{pmatrix}
    1-q_1 & q_1 \\ 1-q_2 & q_2 \end{pmatrix}^{-1}\begin{pmatrix} \calN(q_1)[\rho]  \nonumber\\ \calN(q_2)[\rho]\end{pmatrix} \\
    &= \frac{1}{q_1 - q_2}\left(\calN(q_1)[\rho] - \calN(q_2)[\rho]\right).  \nonumber
\end{align}
\end{proof}
The proof technique extends analogously to channels, which are convex combinations of more than two channels.

\emph{Proof of Observation \ref{obs:unitary_parameter_shift}.}
We will now provide a proof of the parameter-shift rule for unitary gates for the sake of completeness.
The derivative of the channel output is given by
\begin{align}
    \partial_{\mu}\calU(\mu)[\rho] = -iG\e^{-i \mu G} \rho \e^{i \mu G} + \e^{-i \mu G} \rho \e^{i \mu G}iG = -i G\calU(\mu)[\rho] + \text{h.c.}
\end{align}
We now proceed as in Ref.~\cite{schuld2019evaluating} and use the identity
\begin{align}
    BQC^{\dagger} + \text{h.c.} = \frac{1}{2}\left[ (B+C)Q(B+C)^{\dagger} - (B-C)Q(B-C)^{\dagger}  \right]
\end{align}
where $B,Q,C$ are matrices and $Q$ is assumed Hermitian. We assume without loss of generality that the eigenvalues of the generator $G$ are $\pm r$, as a rescaling can always be done at the cost of an irrelevant global phase and such a rescaling leaves the difference of the eigenvalues intact. We define $r = |\lambda_1 - \lambda_2|/2$ and take $Q = \calU(\mu)[\rho]$, $B = -i r^{-1} G$, and $C = \Id$ to obtain
\begin{align}
    \partial_{\mu}\calU(\mu)[\rho] &= \frac{r}{2}\left[ (-i r^{-1} G+\Id)\calU(\mu)[\rho](-i r^{-1} G+\Id)^{\dagger} - (-i r^{-1} G-\Id)\calU(\mu)[\rho](-i r^{-1} G-\Id)^{\dagger}  \right] \\
    &=\frac{r}{2}\left[ (\Id - i r^{-1} G)\calU(\mu)[\rho](\Id - i r^{-1} G)^{\dagger} - (\Id + i r^{-1} G)\calU(\mu)[\rho](\Id + i r^{-1} G)^{\dagger} 
    \right]. \nonumber
\end{align}
Now, we make use of Ref.\ \cite[Theorem 1]{schuld2019evaluating}, which states
\begin{align}
    \calU\left(\frac{\pi}{4r}\right)[\rho] = \frac{1}{2}(\Id - i r^{-1} G)\rho(\Id - i r^{-1} G)^{\dagger},
\end{align}
to arrive at the desired result
\begin{align}
    \partial_{\mu}\calU(\mu)[\rho] &= r \left[\left(\calU\left(\frac{\pi}{4r}\right) \circ \calU(\mu)\right)[\rho] - \left(\calU\left(-\frac{\pi}{4r}\right) \circ \calU(\mu)\right)[\rho]\right]\\
    &= r \left[\calU\left(\mu+\frac{\pi}{4r}\right)[\rho] - \calU\left(\mu-\frac{\pi}{4r}\right)[\rho]\right], \nonumber
\end{align}
where we have exploited the composition rule for gates generated by the same generator. 

\section{Conditions for the parameter-shift rule}\label{sec:app_param_shift_lab}
We will now detail when we can evaluate the derivatives of the output probabilities with respect to the encoded parameters on the sensing platforms directly or by emulation.
The first necessary ingredient is control over the measured parameters, as we need to \enquote{shift} them.
This means we have to be able to re-run the evolution with the parameter in question shifted by a fixed amount. This is feasible if we emulate the parameter encoding, but usually impossible on the actual sensing platform.
However, if we are able to synthesize the evolution under the generators $\calH_j$, we can circumvent this obstacle by adding the additional shift \emph{before} the parameter encoding. This can for example be done in optics experiments, where additional phase shifts on the modes undergoing the parameter encoding can be implemented.
The noise in the state preparation, the parameter encoding and the POVM is taken care of by the approach outlined in this work, but an imperfect realization of the parameter-shift itself will still introduce some noise in the optimization process. 

Given control over the measured parameters, we can use finite differences to obtain an estimate of the derivative. Superior estimates are, however, available if the encoding unitary admits a parameter-shift rule. 
In discrete variables, this is the case if the parameters $\phi_j$ are encoded by separate evolutions under Hamiltonians of the form $g_j(\pphi) H_j$ where $H_j$ only has two distinct eigenvalues. We can then always reparametrize the problem in terms of the new variables $g_j$ and use the transformation rule of the CFIM to obtain the CFIM with respect to the original parameters $\pphi$. 
In the Gaussian case, all elementary operations (phase shift, squeezing and displacement) admit a parameter-shift rule. If the parameter is encoded by separate operations of this kind, the parameter-shift rule applies.

Let $d$ be the number of metrology parameters. With either method, we need to evaluate the setup at $2d + 1$ points in parameter space to calculate the CFIM. Recent work \cite{banchi2020measuring} has extended the reach of parameter-shift like techniques by giving a method to obtain analytic gradients via stochastic methods. These can be applied to more intricate evolutions than the regular parameter-shift rule and even extend beyond the noisy unitary model in certain cases, at the cost of having to evaluate the setup at more points in parameter space. 

\section{Sensing emulation}\label{sec:app_emulation}
To emulate a sensing platform, we need to reproduce the evolutions happening in the state preparation, the POVM, and the occurring error channels. Realizing the evolutions is simple, as all unitary evolutions can be decomposed into elementary gates. The realization of error channels, on the other hand, can proceed in two ways.

First, the Stinespring dilation theorem \cite{wilde2013quantum} guarantees that all error channels can be realized by a unitary evolution on a larger system. We can thus use ancillary qubits or modes to emulate an error channel.
Second, we can emulate any mixed-unitary (also known as random-unitary) quantum channel on a parametrized quantum circuit. In such a channel, a unitary is picked from a set $\{ U_k \}$ with probability $p_k$. It can be emulated by executing the quantum circuit multiple times while repeatedly sampling random instances of the unitary. The mixed-unitary channels comprise important examples like bit-flip channels, depolarizing channels and dephasing channels. 
Viewing NISQ devices as parametrized quantum circuits is usually abstracted far from the hardware. More complicated error channels than mixed-unitary channels can be engineered when closer access to the hardware is possible, but the capabilities naturally depend strongly on the underlying technology.

\section{Extensions}\label{sec:app_ext}
The variational approach outlined in the main text can be extended in various ways. In the following, We detail three important extensions.

\subsection{Bayesian quantum metrology: Including prior knowledge}\label{sec:ext_bayesian}
The Cram{\'e}r-Rao bound \eqref{eqn:crb} is a function of the underlying parameters $\pphi$. If one has prior knowledge about their distribution $P(\pphi)$, one can derive a more general bound, the \emph{Van-Trees inequality} \cite{van2013detection}
\begin{align}\label{eqn:van_trees}
    \operatornamewithlimits{Cov}[\hat{\vvarphi}] \geq [n \bbE_P\{ I_{\pphi}\} + I_{P}]^{-1}.
\end{align}
Compared to the Cram{\'e}r-Rao inequality, we now use the expected CFIM over the distribution $P(\pphi)$ and have to add the CFIM of the prior $P(\pphi)$ itself. But only the expected CFIM scales with the number of samples $n$. We can therefore neglect the contribution of the prior as long as we are not considering sample-frugal estimation. The construction of the cost function proceeds as in the main text by applying a positive semi-definite weighting matrix to both sides of Eq.~\eqref{eqn:van_trees} and using the right hand side as the cost function.

If the variational approach is directly applied to the sensing platform, the Bayesian nature of the underlying parameters is automatically taken into account. If the sensing is simulated on a quantum device, we can include the prior distribution by performing multiple runs of the experiment where $\pphi$ is freshly pulled from $P(\pphi)$ during each run. The CFIM can then be estimated from the results of the different runs.

\subsection{Including time dependence}\label{sec:ext_time_dependence}
While the noisy unitary model is quite general, in many relevant cases both the encoded parameters and the noise depend on the exposure time $t$ to the physical encoding process.
Due to the interdependence of noise and unitary evolution, we can no longer use the parameter-shift rule. But we can still resort to finite differences to obtain the derivatives, both on the sensing platform and in the emulated context. 
In the simulation part, we have the additional possibility to compute the derivative directly. We will now assume that both the POVM (via the error channel) and the noise-free state depend on  $t$. The time-derivative of the probabilities is then
\begin{align}
    \partial_t p_l = \Tr\{ (\partial_t \Pi_l') \rho \} + \Tr \{ \Pi_l' (\partial_t \rho) \}.
\end{align}
To calculate the channel-dependent part, we need the time-derivative of the error channel on the classical side of the optimization
\begin{align}
    \frac{\partial}{\partial t}\calN^{\dagger}[\Pi_l] = \dot{\Pi}_l'.
\end{align}
For the calculation of the state-dependent part, we can re-use the output probabilities we already calculated
\begin{align}\label{eqn:time_dependence_chain_rule}
    \frac{\partial}{\partial t}\rho(\pphi) = \sum_j \frac{\partial \phi_j}{\partial t} \partial_j \rho(\pphi).
\end{align}
If the parameter-encoding is parameter-shift differentiable, this is already sufficient. If not, we also have to take care of the time-derivatives of the partial derivatives. The desired quantity is
\begin{align}
    \partial_t \partial_j p_l &= \partial_t \Tr \{ \Delta_l^j \rho \} \\
    &= \Tr \{ (\partial_t \Delta_l^j) \rho \} + \Tr \{ \Delta_l^j (\partial_t \rho) \}.
     \nonumber
\end{align}
The first term encodes the error dependence, we analyze it using the definition of the derivative operators \eqref{eqn:def_derivative_ops}
\begin{align}
    \Tr \{ (\partial_t \Delta_l^j) \rho \} &= -i \Tr \{ (\partial_t [\Pi_l', \calH_j]) \rho \}
    = -i \Tr \{ ( [\partial_t\Pi_l', \calH_j]) \rho \}  \\
    &= -i \Tr \{ [\dot{\Pi}_l', \calH_j] \rho \}
    = \Tr \{ \dot{\Delta}_l^j \rho \},
     \nonumber
\end{align}
where we have defined 
\begin{align}
    \dot{\Delta}_l^j \coloneqq -i [\dot{\Pi}_l', \calH_j].
\end{align}
The second term encodes the dependence of the encoded state. We can analyze it by using Eq.~\eqref{eqn:time_dependence_chain_rule} and the differentiation rule Eq. \eqref{eqn:derivative_h_ops}:
\begin{align}
    \Tr \{ \Delta_l^j (\partial_t \rho) \} &= \sum_j \frac{\partial \phi_j}{\partial t} \Tr \{ \Delta_l^j  \partial_j \rho \} \\
    &=-i \sum_j \frac{\partial \phi_j}{\partial t} \Tr \{ \Delta_l^j \calH_j\rho - \Delta_l^j \rho \calH_j\}  \nonumber\\
    &=-i \sum_j \frac{\partial \phi_j}{\partial t} \Tr \{ \Delta_l^j \calH_j\rho - \calH_j\Delta_l^j \rho \} \nonumber \\
    &=-i \sum_j \frac{\partial \phi_j}{\partial t} \Tr \{ [\Delta_l^j, \calH_j] \rho\}  \nonumber \\
    &= \Tr \{ \ddot{\Delta}_l^j\rho\},  \nonumber
\end{align}
where we have defined
\begin{align}
    \ddot{\Delta}_l^j &\coloneqq -i \sum_j \frac{\partial \phi_j}{\partial t} [\Delta_l^j, \calH_j] \\
    &= - \sum_j \frac{\partial \phi_j}{\partial t} [[\Pi_l', \calH_j], \calH_j].
     \nonumber
\end{align}
With of the above definitions we arrive at
\begin{align}
    \partial_t \partial_j p_l &= \Tr \{ \dot{\Delta}_l^j \rho \} + \Tr \{ \ddot{\Delta}_l^j \rho \}.
\end{align}
We now have again reformulated all derivatives as expectation values at the cost of more complicated expressions. The derivatives can then be used in conjunction with Eq.~\eqref{eqn:cfim} to obtain the CFIM and subsequently the cost function.

\subsection{Evaluation of the quantum Cram{\'e}r-Rao bound}\label{sec:ext_qcrb}
Helstrom and Holevo \cite{helstrom1976quantum,holevo1982probabilistic} have shown that there exists a more fundamental bound than the Cram{\'e}r-Rao bound that is independent of the chosen POVM, the \emph{Quantum Cram{\'e}r-Rao bound (QCRB)}
\begin{align}\label{eqn:qcrb}
    \operatornamewithlimits{Cov}[\hat{\vvarphi}] \geq \frac{1}{n} \calF_{\ff}^{-1},
\end{align}
where $\calF_{\pphi}$ is the \emph{quantum Fisher information matrix (QFIM)} defined as
\begin{align}\label{eqn:qfim_def}
    [\calF_{\pphi}]_{j,k} &\coloneqq  \frac{1}{2} \Tr \{ \rho(\pphi) \{L_j, L_k \} \},
\end{align}
where the symmetric logarithmic derivative (SLD) operators $L_j$ are implicitly defined by the equation
\begin{align}\label{eqn:def_sld}
    \frac{\partial}{\partial \phi_j} \rho(\pphi) = \frac{1}{2}\{\rho(\pphi), L_j\}.
\end{align}
The properties of the QFIM and its various applications in quantum metrology are discussed in detail in Ref.~\cite{liu2019quantum}.

It would be satisfying to also be able to evaluate the fundamental precision limit for the given sensing problem that is given by the QCRB. Nonetheless, the calculation of the Quantum Cram{\'e}r-Rao bound is much more intricate than the calculation of the classical Cram{\'e}r-Rao bound. Formulas exist for unitary encoding processes, but calculating the QFIM for an arbitrary mixed probe state requires either full tomography or solving an infeasibly large semidefinite program \cite{katariya2020geometric}. 

For a pure probe state on the other hand, the calculation of the QFIM is straightforward, as it reduces to the fourfold covariance matrix of the generators \cite{liu2019quantum}
\begin{align}\label{eqn:qfim_pure_cov}
    [\calF_{\pphi}]_{j,k} &= 4 \operatorname{Cov}_{\ket{\psi}}(\calH_j, \calH_k) 
\end{align}
evaluated on the final state vector $\ket{\psi(\ttheta, \pphi)}$. It can thus be estimated from measurements of the observables $\{ \calH_j \}$ and their products and therefore also be differentiated via the parameter-shift rule. We can thus construct a similar cost function to Eq.~\eqref{eqn:cost_function}, replacing the CFIM with the QFIM, and carry out an optimization. Current quantum devices usually do not reach the purity required to perform this optimization, so we have to resort to classical simulation instead.

Convexity of the quantum Fisher information ensures that the maximum Fisher information is attained on pure states. We can thus expect the pure picture to capture the optimal sensing performance when all noise is suppressed. This can serve as a benchmark to judge the performance of the probes obtained by the methods outlined above.

\section{Numerical experiments}\label{sec:app_numerics}

\subsection{Ramsay spectroscopy}
\emph{Circuit.}
The following quantum circuit was used to perform the numerical simulations:
\begin{center}
     \includegraphics[width=.5\columnwidth]{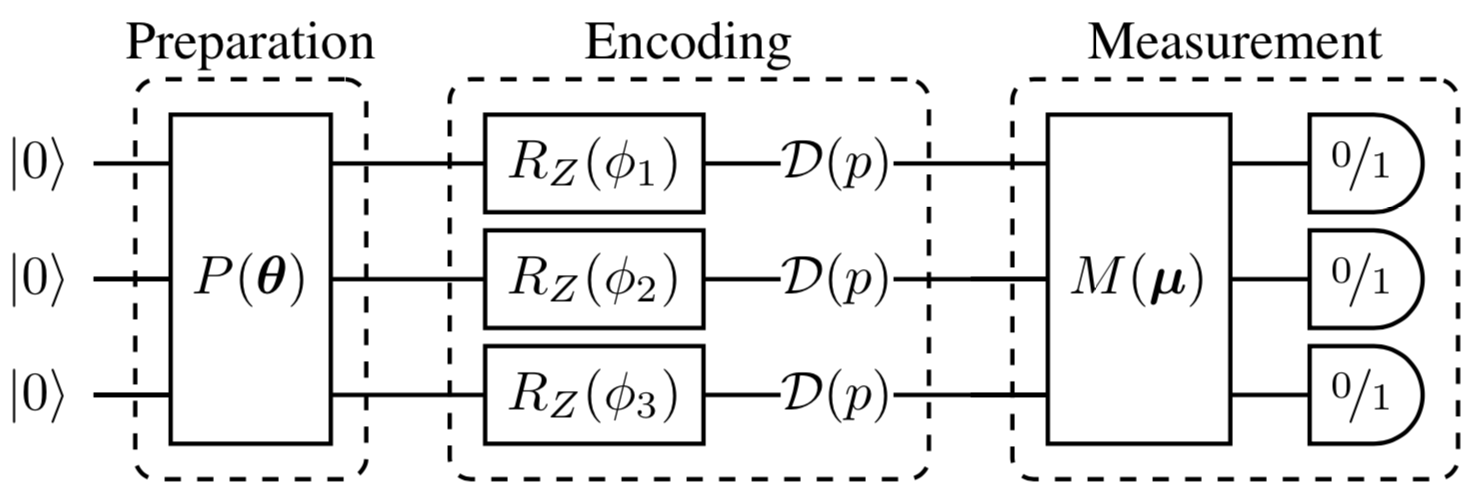}
\end{center}
The preparation and measurement circuits considered are detailed below. The encoding models the phase shift picked up during the interaction with additional dephasing noise $\calD$ with dephasing probability $p$.
We consider the following preparation circuit that parametrizes any pure three qubit state with the minimal number of parameters:
\begin{center}
 \includegraphics[width=.85\columnwidth]{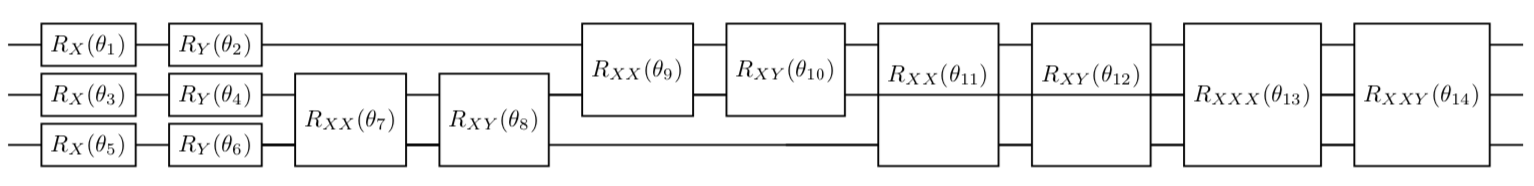}
\end{center}
Here, $R_P(\theta)$ is a rotation generated by the Pauli word $P$.
The measurement is parametrized by a local unitary transformation:
\begin{center}
     \includegraphics[width=.5\columnwidth]{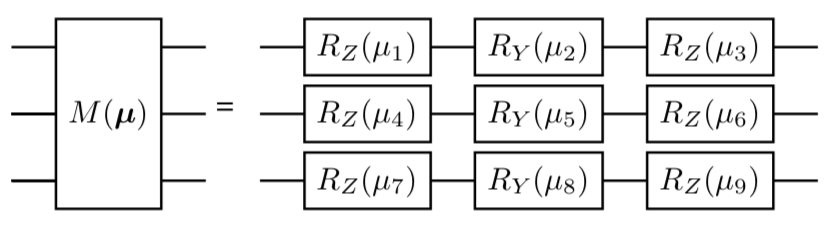}
\end{center}

\emph{Optimization.}
The encoded phases were fixed at $\phi_j = {\pi}/{6}$, the optimal parameters for a GHZ probe~\cite{huelga1997improvement}. The dephasing probability of the encoding was varied in the interval $p \in [0, \num{0.5}]$. The weights of the variational circuit was initialized at random, with each weight $\theta_j$, $\mu_k$ randomly drawn from a uniform distribution over $[0,2\pi]$. The optimization was executed using gradient descent for 1000 steps with an initial step size of \num{0.05}. A regularization term of $\num{1e-10} \times \bbI$ was added to the computed CFIM before inversion during the calculation of the cost function to avoid inverting a singular matrix. Multiple runs of the above optimization were repeated and the best results selected.

\subsection{NV trilateration}
We consider a setup with three negatively charged NV centers (NV1, NV2, NV3) located in a diamond sample at positions $\boldsymbol{r}_j=(x_j, y_j, z_j)$. We assume that these positions are known due to prior characterization of the system. The sensing task is to locate the position of a target spin J at an unknown position $\boldsymbol{r}=(x,y,z)$. As in Ref.~\cite{sushkov2014magnetic}, we are using the $m_s=0$ and $m_s=-1$ ground states of the NV centers for the sensing task.

We consider the interactions between the NV centers and the target spin to be governed by the following dipole interaction Hamiltonian under the secular approximation~\cite{sushkov2014magnetic}
\begin{align}
    H = \hbar \Delta \sum_j (S_j^n)^2 + \hbar \gamma_e \sum_j (\mathbf{B}\mathbf{n}_j) S_j^n + \hbar \gamma_e B J^b + \sum_j \frac{\hbar^2 \mu_0 \gamma_e^2}{4\pi \delta r_j^3}S_j^n J^b [\mathbf{n}_j \mathbf{b} - (\mathbf{n}_j \mathbf{e}_j)(\mathbf{b} \mathbf{e}_j)]
\end{align}
where $S_n^f$ is the spin operator of the $j$'th sensing spin projected on its symmetry axis, which has unit vector $\mathbf{n}_j$. The spin to be sensed (target spin) has spin operator $J^b$ projected on the magnetic field axis ($\mathbf{B}$ with unit vector $\mathbf{b}$). The distance between the $j$'th sensing spin and the target spin is $\delta r_j$ and has corresponding unit vector $\mathbf{e}_j$. The parameters $\Delta,\gamma_e,$ and $\mu_e$ are the zero field splitting of the NV ground states, the electron's gyromagnetic ratio, and the magnetic constant, respectively. Note that we have neglected interaction between the sensing spins, assuming that they can be decoupled by suitable pulse sequences~\cite{Abobeih2019}. Performing the DEER sequence, the Hamiltonian simplifies to
\begin{align}
    H_{\text{opt}} &= \sum_j \left( \frac{\hbar^2 \mu_0 \gamma_e^2}{4\pi \delta r_j^3}S_j^n J^b [\boldsymbol{n}_j \boldsymbol{b} - (\boldsymbol{n}_j \boldsymbol{e}_j)(\boldsymbol{b} \boldsymbol{e}_j)] \right),
\end{align}
where we have neglected the local evolution of the NV centers ($\sum_j \hbar \Delta (S_j^n)^2$) since this only adds a constant phase that can be mitigated.

From the Hamiltonian, we can identify the independent phases
\begin{align}
    H= \sum_j \phi_j(\rr) S_j^n,
\end{align}
where
\begin{align}
    \phi_j(\rr) =  \frac{C}{\delta r_j^3} [\boldsymbol{n}_j \boldsymbol{b} - (\boldsymbol{n}_j \boldsymbol{e}_j)(\boldsymbol{b} \boldsymbol{e}_j)],
\end{align}
and we have absorbed all constants into $C$. With this expression at hand, we use the autodifferentiation features of sympy~\cite{sympy} to calculate the Jacobian $J_{j,k} = {\partial \phi_j}/{\partial r_k}$. The Jacobian necessary for our optimization can then be calculated by inverting $J$.

\emph{Circuit.}
The following quantum circuit is used to perform the numerical simulations:
\begin{center}
     \includegraphics[width=.5\columnwidth]{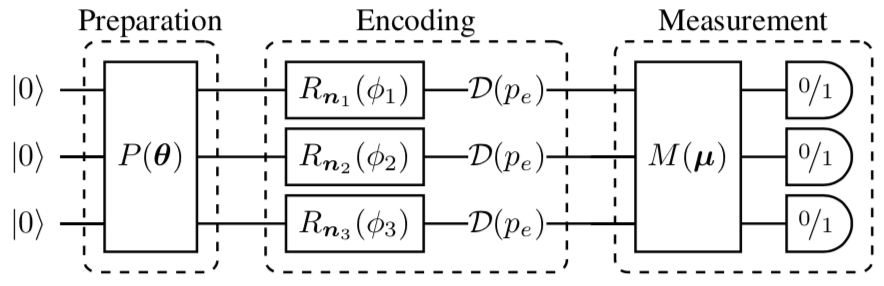}
\end{center}
The preparation and measurement circuits considered are detailed below. The encoding models the evolution under the aforementioned Hamiltonian with additional dephasing noise $\calD$ with dephasing probability $p_e$.

We seek to study the trade-off between a shallow entangled probe and a local probe at different gate noise levels. The gate noise will be modeled as depolarizing noise $\calP$ with depolarization probability $p_g$ that is equal for single- and two-qubit gates.
We consider two different state preparations which use the following primitives:
First, an arbitrary single qubit operation
\begin{center}
     \includegraphics[width=.5\columnwidth]{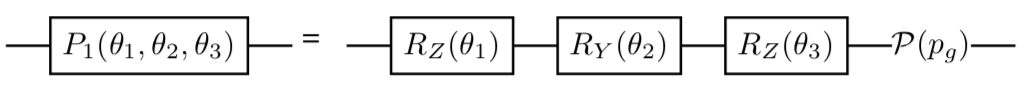}
\end{center}
Second, an entangling two-qubit operation
\begin{center}
     \includegraphics[width=.4\columnwidth]{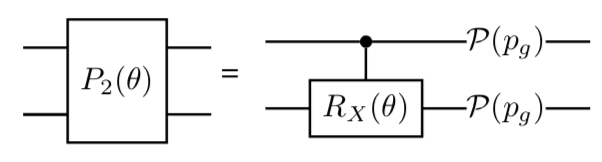}
\end{center}
The first preparation we consider is a simple local state preparation
\begin{center}
     \includegraphics[width=.33\columnwidth]{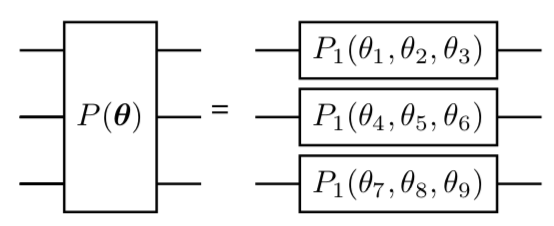}
\end{center}
The comparison is made with a shallow entangling state preparation that extends the local state preparation by executing two entangling operations
\begin{center}
     \includegraphics[width=.5\columnwidth]{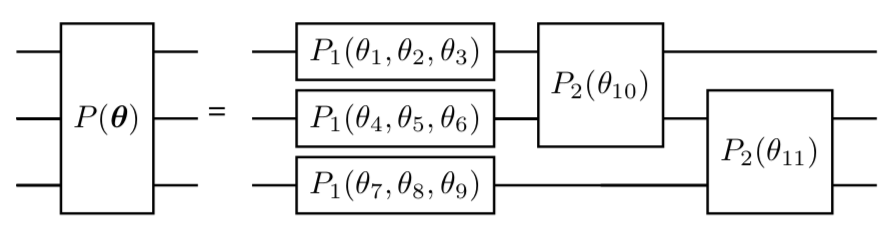}
\end{center}
We note that this state preparation can create a GHZ state.
For both preparations, we use a local unitary to parametrize the measurement
\begin{center}
     \includegraphics[width=.35\columnwidth]{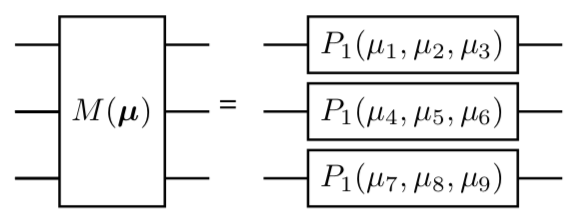}
\end{center}

\emph{Optimization.}
\begin{table}\label{tab:nv_optimization_data}\centering
\begin{tabular}{@{}llcccc@{}}\toprule
& & NV\textsubscript{1} & NV\textsubscript{2} & NV\textsubscript{3} & Target Spin \\ \midrule
Location & $\rr$ & $(0,0,0) $ & $(1,2,3)$ & $(4,1,1)$ & $(2,0,2)$ \\
Symmetry Axis & $\operatorname{arg}(\nn)$ & $(\num{0.3} \pi, \num{0.25} \pi) $ & $(\num{0.7} \pi, \num{-0.15} \pi)$ & $(\num{-0.2} \pi, \num{1.2} \pi)$ & $(0,0)$ \\
\bottomrule
\end{tabular}
\caption{Ground truth used for the NV trilateration experiment. The location is given in unit-free coordinates, the symmetry axis is given in spherical coordinates. Note that the symmetry axis in the case of the target spin is the precession axis identical to the magnetic field direction.}
\end{table}
The ground truth used for the numerical experiment is shown in Tab.~I. The evolution time was fixed so that the total accumulated phase $\phi_1 + \phi_2 + \phi_3 = \pi/2$. The dephasing probability of the encoding was fixed at $p_e = \num{0.1}$ while the gate depolarization probability was varied in the interval $p_g \in [0, \num{0.1}]$.

The weights of the variational circuits have been initialized at random, with each weight $\theta_j$, $\mu_k$ randomly drawn from a uniform distribution over $[0,2\pi]$. The optimization was executed using gradient descent for 1000 steps with an initial step size of \num{0.01}. To avoid issues with overly large step sizes we used the following strategy: every time an increase of the cost function after a step was detected, the step size was multiplied with a factor of \num{0.7} and the step repeated. A regularization term of $\num{1e-10} \times \bbI$ was added to the computed CFIM before inversion during the calculation of the cost function to avoid inverting a singular matrix. Multiple runs of the above optimization were repeated and the best results selected.

\section{CO\textsubscript{2} Emission Table}\label{sec:appconduct}
\begin{table}[h]
\label{tab:cotwo}
\begin{tabular}[b]{l c}
\toprule
\textbf{Numerical simulations} & \\
\midrule
Total Kernel Hours [$\mathrm{h}$]& 29071 \\
Thermal Design Power Per Kernel [$\mathrm{W}$]& 5.75\\
Total Energy Consumption Simulations [$\mathrm{kWh}$] & 167.2 \\
Average Emission Of CO$_2$ In Germany [$\mathrm{kg/kWh}$]& 0.56\\
Total CO$_2$ Emission For Numerical Simulations [$\mathrm{kg}$] & 94 \\
\midrule
\textbf{Transport} & \\
\midrule
Total CO$_2$ Emission For Transport [$\mathrm{kg}$] & 0\\
\midrule
Total CO$_2$ Emission [$\mathrm{kg}$] & 94 \\
Were The Emissions Offset? & Yes \\
\bottomrule
\end{tabular}
\end{table}

\end{document}